\documentclass{article}  

\usepackage{booktabs}

\usepackage{macro}
\usepackage{listings}
\usepackage{comment}
\usepackage{soul,color}
\usepackage{changebar}

\usepackage[subtle]{savetrees}

\newcommand{\CGS}{\texttt{vCGS}}
\newcommand{\CGSname}{\texttt{M}}

\renewcommand{\top}{\mathrm{tt}}
\renewcommand{\bot}{\mathrm{ff}}

\newcommand{\naww}[1]{{\langle\!\langle #1 \rangle\!\rangle}}

\def\qed{\hfill{\qedboxempty}      
	\ifdim\lastskip<\medskipamount \removelastskip\penalty55\medskip\fi}

\def\qedboxempty{\vbox{\hrule\hbox{\vrule\kern3pt
			\vbox{\kern3pt\kern3pt}\kern3pt\vrule}\hrule}}

\newtheorem{theorem}{Theorem}

\newtheorem{corollary}{Corollary}

\newtheorem{proof}{Proof.}

\begin{document}

\title{Model Checking $ATL^*$ on \CGS}



\author{F. Belardinelli, C. Dima, I. Boureanu, and V. Malvone}

\maketitle


%
%
%
%
%
%
%
%
%
%

We prove that the model checking $ATL^*$ on \CGS\ is
undecidable.  To do so, we reduce this problem to model checking $ATL^*$
on iCGS.
%

Consider an $ATL^*$ formula $\varphi$ to be model-checked
on a given iCGS $\CGSname = \langle Ag, \{
Act_a \}_{a \in Ag}, S, S_0, P, \tau, \{\sim_a \}_{a \in Ag}, \pi
\rangle$.  Out of $\CGSname$ and $\varphi$, we define an
\CGS\ $\Delta_{\CGSname}$ (which we sometimes refer to simply as $\Delta$)
and an $ATL^*$ formula $\psi_{\varphi,\CGSname}$, as follows.

\underline{Agents \& atoms}. The set of agents in $\Delta_{\CGSname}$ is
$Ag'= Ag \cup\{e\}$, where $e$ denotes the \emph{Environment} agent.
For each agent $a \in Ag$, the set of atoms controlled by $a$ includes
an atom for each action in $Act_a$, i.e., $V_a = \{p_{act} \mid act
\in Act_a\}$.  The environment controls atoms corresponding to each
state of the iCGS and each indistinguishability class for each
agent, i.e., $V_e = \{p_s \mid s \in S\} \cup \{p_{[s]_a } \mid s \in
S, a \in Ag\}$, with $[s]_a = \{ s' \in S \mid s' \sim_a s\}$.

An agent specification is $spec_a = \langle V_a, GC_a \rangle$ with
the guarded commands to be defined hereafter.
First, for each agent specification $spec_a$ and subset $W \subseteq
V_a$, by \emph{$invis^a(W)$} we denote the boolean formula specifying
that all of $a$'s atoms in $W$ are left invisible to any agent other
than $a$, i.e., $a$ does not share any of her controlled atoms from
$W$ with anyone:
\[
invis^a(W) = \bigwedge_{b \in Ag \setminus{\{a\}}} \bigwedge_{v\in W}  vis(v,b) := \bot
\]

For the special case of the environment (i.e., where $a$ in the
above is replaced with the Environment and $W \subseteq V_e$,
etc), we simply write \emph{$invis(W)$} instead of $invis^a(W)$ (i.e.,
$invis(W) = \bigwedge_{b \in Ag} \allowbreak \bigwedge_{v\in W} vis(v,b)  := \bot$.)

Second, for an agent $b \in Ag \setminus{\{a\}}$, by \emph{$vis(W,b)$}
we denote the boolean formula specifying that all atoms in $W$ are
visible to $b$, i.e., $a$ does share all her controlled atoms in $W$
with $b$:
\[
vis(W,b) = \bigwedge_{v\in W} vis(v,b) = \top
\]

For each agent $a\in Ag$ and for the Environment $e$, we also use boolean variables $turn_a$ and $turn_e$ to simulate their turns and mechanise the (synchronisation over) actions.

\underline{Guarded commands of $\texttt{init}$-type.}
\begin{enumerate}
\item For each $s_0 \in  S_0$,  a guarded command $\gamma_{[s_0]_a}$ of $\texttt{init}$-type to agents $a$  is defined in $\Delta_\CGSname$ as follows:
\begin{tabbing}
$\gamma_{[s_0]_a}$ \= $::=$ \=  $p_{[s_0]_a} \rightsquigarrow turn_a = \bot, invis(V_a)$       
\end{tabbing}
\item For each $s_0 \in S_0$, a guarded command $\gamma_{[s_0]_e}$ of
  $\texttt{init}$-type for the Environment $e$ is defined in
  $\Delta_\CGSname$ as follows:
\begin{tabbing}
$\gamma_{0_e}$\= $::=$ \=$\top \!\rightsquigarrow\! invis(\{p_{s_0} \!\mid\! s_0 \in S_0\}), \bigwedge_{b\in Ag}\!\! vis(\{p_{[s_0]_b}\},b)$         
\end{tabbing}
\end{enumerate}

\underline{Guarded commands of $\texttt{update}$-type.}
Let $Act_a$ in $\CGSname$ be given as the set $\{\alpha_1,\ldots,
\alpha_{k_a}\}$.  Then, agent $a$'s guarded commands of
$\texttt{update}$-type are added in $\Delta_\CGSname$ as follows:
\begin{enumerate}
\setcounter{enumi}{2}
\item For every $1\leq i\leq k_a$ and each equivalence class $[s]_a\in
  S$, guarded command $\gamma_{i,[s]_a}$ is defined as
\begin{tabbing}
$\gamma_{i,[s]_a}$ \= $::=$ \= $turn_a \wedge p_{[s]_a}
  \rightsquigarrow$ \\ \> \> $p_{\alpha_i} := \top,  \bigwedge_{j\neq i} p_{\alpha_j} := \bot, turn := \bot$
\end{tabbing}
\item To simulate $a$'s turn (or $a$ ``moving forward'' in the system --hence the name below),   we add the following guarded command:
\begin{tabbing}
$\gamma_{fwd_a}$ \= := \=  $\neg turn \rightsquigarrow turn := \top$ 
\end{tabbing}
\end{enumerate}

Then, the environment agent has the following guarded commands of
$\texttt{update}$-type:
\begin{enumerate}
\setcounter{enumi}{4}
\item One guarded command of $\texttt{update}$-type for $e$, to make him ``move forward'' (i.e., take turns):
\begin{tabbing}
$\gamma_{fwd_e}$ \= $::=$ \= $\neg turn_e \rightsquigarrow turn_e := \top$
\end{tabbing}
\item For each $s \in S$, $a \in Ag$, and $i_a \leq k_a$, let $t$
  denoted $\tau (s, (\alpha_{i_a})_{a\in Ag})$; for this, another
  guarded command for $e$ is as follows:

\begin{tabbing}
$\gamma_{s,(i_a)_{a\in Ag}}$ \= $::=$ \= $turn_e \wedge p_s \wedge p_{\alpha_{i_a}} \rightsquigarrow$  \\ \> \>
    $turn_e := \bot, p_t := \top, \bigwedge_{b\in Ag} p_{[t]_b} := \top,$  \\ \> \>
    $\bigwedge_{p \in\pi(s) } p := \top, \bigwedge_{p \not \in \pi(s)} p := \bot$ \\ \> \>
     $\bigwedge_{u \in S\setminus \{t\}} \Big( p_u = \bot, \bigwedge_{b\in Ag} p_{[u]_b} = \bot\Big)$
\end{tabbing}
\end{enumerate}

Now, using the reduction above (which is in PTIME), we can formally state
the following result:
\begin{theorem}[3.4] \label{th1}
The model checking problem for $ATL^*$ (resp.~$ATL$) on iCGS is
  PTIME-reducible to the same on \CGS.
\end{theorem}
\begin{proof}
Given an $ATL^*$ formula $\varphi$, we construct the formula $\varphi'$ in
which each next operator is duplicated. For example, for
the formula $\varphi = X p$, we set $\varphi' = X X
 p$.
 For $ATL$ we duplicate each coalition-next operator instead: for
 $\varphi = \naww{A} X p$, we set $\varphi' = \naww{A} X
\naww{A} X p$.
Then we can show that, for any iCGS $\CGSname$, the \CGS\ $\Delta_\CGSname$
constructed as above is such that $\CGSname \models \varphi$ iff $\Delta_\CGSname
\models \varphi'$. We do so by structural induction on the formula
$\varphi$.
\end{proof}

Using Theorem~\ref{th1} above and knowing from~\cite{DimaT11} that
model checking $ATL^*$ and $ATL$ on iCGSs is undecidable,
we can state
the following result:
\begin{corollary}[3.5]
 The model checking problem for $ATL^*$ (resp.~$ATL$) on \CGS\ are
 undecidable.
\end{corollary}

We conclude by recalling that if we assume {\em positional
  strategies}, the model checking problem for $ATL^*$ (resp.~$ATL$) on
iCGS are PSPACE- (resp.~$\Delta^P_2$-) complete
\cite{Schobbens04,JamrogaDix06}. Hence, the same complexities apply
to \CGS.

%
%


\bibliographystyle{abbrv}  

\end{document}